\pdfoutput=1
\documentclass{article}
\pdfoutput=1
\usepackage{ijcai19}
\usepackage{amsthm}
\usepackage[english]{babel}
\newtheorem{theorem}{Theorem}

\newtheorem{proposition}[theorem]{Proposition}

\usepackage{amsthm}
\usepackage{times}
\usepackage{soul}
\usepackage{url}
\usepackage[hidelinks]{hyperref}
\usepackage{graphicx}
\usepackage{amsmath}
\usepackage{booktabs}
\usepackage{subfigure}
\usepackage{multicol}  
\usepackage{eucal}
\usepackage{enumitem}
\usepackage{multirow}
\usepackage{algorithm}
\usepackage{algorithmic}
\usepackage{booktabs} 
\usepackage{multicol}  
\usepackage{multirow}
\usepackage{marvosym}
\usepackage{amssymb}
\usepackage{subfigure}
\usepackage{amsmath}
\usepackage{algorithm}
\usepackage{algorithmic}
\usepackage{physics}
\usepackage{color}

\DeclareMathOperator*{\argmax}{arg\,max}

\title{MEGAN: A Generative Adversarial Network for Multi-View Network Embedding}

\author{
Yiwei Sun$^1$\and
Suhang Wang$^2$\and
Tsung-Yu Hsieh$^{1}$\and
Xianfeng Tang$^2$\and
Vasant Honavar$^1{}^2$\\
\affiliations
$^1$Department of Computer Science and Engineering, The Pennsylvania State University, USA\\
$^2$College of Information Sciences and Technology, The Pennsylvania State University, USA\\
\emails
\{yus162,szw494,tuh45,xut10,vuh14\}@psu.edu
}

\begin{document}

\maketitle

\begin{abstract}
Data from many real-world applications can be naturally represented by {\em multi-view} networks where the different views encode different types of relationships (e.g., friendship, shared interests in music, etc.) between real-world individuals or entities. There is an urgent need for methods to obtain low-dimensional, information preserving and typically nonlinear embeddings of such multi-view networks. However, most of the work on multi-view learning focuses on data that lack a network structure, and most of the work on network embeddings has focused primarily on {\em single-view} networks. Against this background, we consider the multi-view network representation learning problem, i.e., the problem of constructing low-dimensional information preserving embeddings of {\em multi-view} networks. Specifically, we investigate a novel Generative Adversarial Network (GAN) framework for Multi-View Network Embedding, namely MEGAN, aimed at preserving the information from the individual network views, while accounting for connectivity across (and hence complementarity of and correlations between) different views. The results of our experiments on two real-world multi-view data sets show that the embeddings obtained using MEGAN outperform the state-of-the-art methods on node classification, link prediction and visualization tasks.
\end{abstract}

\section{Introduction}
Network embedding or network representation learning, which aims to learn low-dimensional, information preserving and typically non-linear representations of networks, has been shown to be useful in many tasks, such as link prediction~\cite{perozzi2014deepwalk}, community detection~\cite{he2015detecting,cavallari2017learning} and node classification~\cite{wang2016linked}. A variety of network embedding schemes have been proposed in the literature  ~\cite{ou2016asymmetric,wang2016structural,grover2016node2vec,hamilton2017inductive,yu2018representation,zhang2018arbitrary}. However, most of these methods focus on {\em single-view} networks, i.e., networks with only one type relationships between nodes.
\begin{figure}[ht!]
      \centering
        \includegraphics[width=8.2cm]{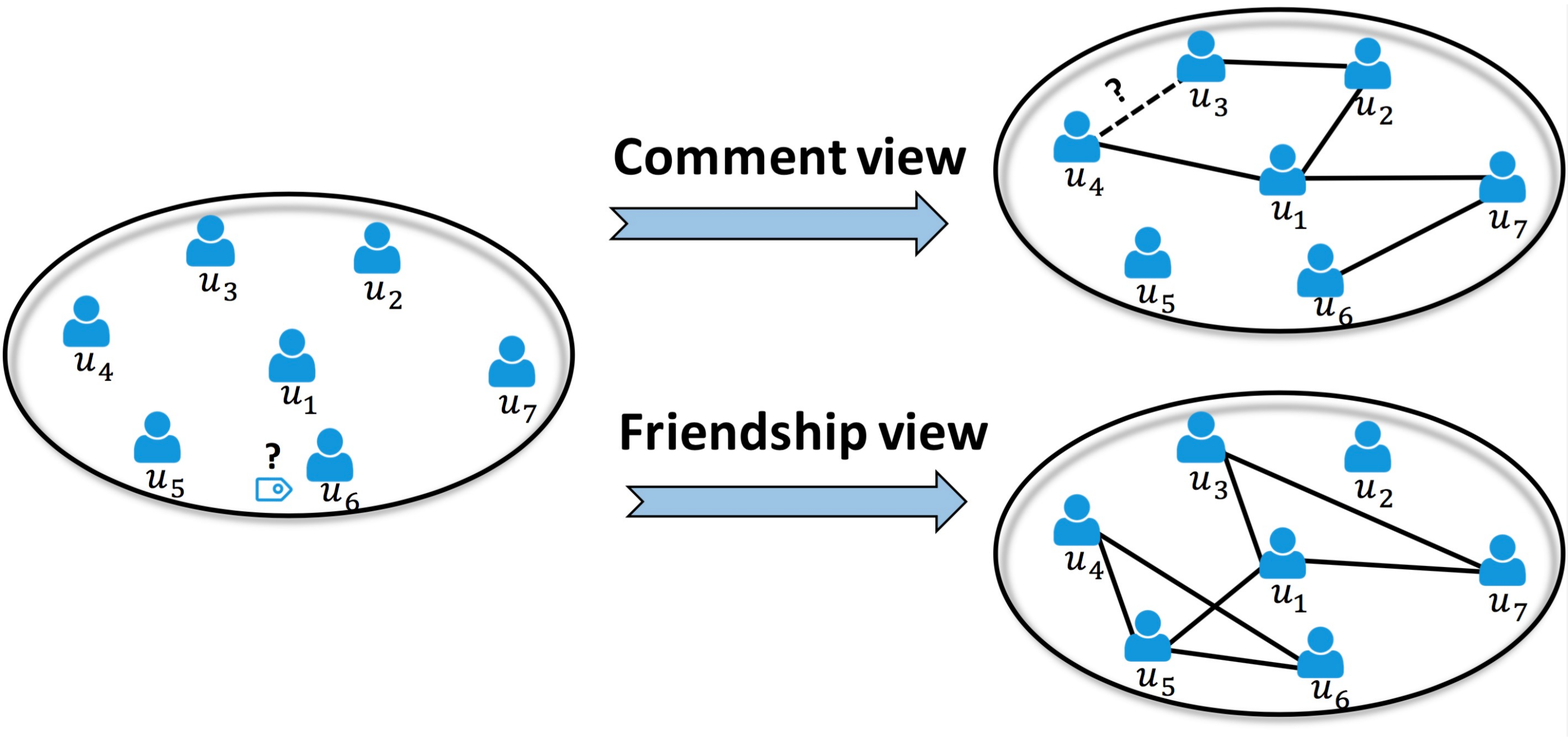}
        \vskip -2.4em
        \caption{The toy multi-view network containing 7 users (nodes) and comprised of comment view and friendship view }
    \vskip -1.5em
       \label{Fig_vis:multiview}%
\end{figure}

However, data from many real-world applications are best represented by {\em multi-view} networks~\cite{kivela2014multilayer}, where the nodes are linked by \emph{multiple} types of relations. For example, in Flickr, two users can have multiple relations such as friendship and communicative interactions (e.g., public comments); In Facebook, two users can share friendship, enrollment in the same university,  or likes and dislikes.  
For example,  Fig.\ref{Fig_vis:multiview} shows a 2-view network comprised of the  friendship view and the comment view. 
Such networks present multi-view counterparts of problems considered in the single-view setting, such as, node classification (e.g., labelling the user to the specific categories of interests(tag of $u_6$)) and link prediction (e.g., predicting a future link, say between $u_3$ and $u_4$). 
Previous work~\cite{wang2015deep,shi2016dynamics} has shown that taking advantage of the complementary and synergy information supplied by the different views can lead to improved performance. Hence, there is a growing interest in multi-view network representation learning (MVNRL) methods that effectively integrate information from disparate views~\cite{shi2016dynamics,qu2017attention,zhang2018scalable,DBLP:conf/uai/HuangLZWYC18,ma2019multi} . However, there is significant room for improving both our understanding of the theoretical underpinnings as well as practical methods on real-world applications. 

Generative adversarial networks (GAN)~\cite{goodfellow2014generative}, which have several attractive properties, including robustness in the face adversarial data samples, noise in the data, etc., have been shown to be especially effective for  modeling the underlying complex data distributions by discovering latent representations. A GAN consists of two sub-networks, a {\em generator} which is trained to generate adversarial data samples by learning a mapping from a latent space to a data distribution of interest, and a discriminator that is trained to discriminate between data samples drawn from the true data distribution and the adversarial samples produced by the generator. Recent work has demonstrated that GAN can be used to effectively perform network representation learning in the single view setting ~\cite{wang2018graphgan,bojchevski2018netgan}. Against this background, it is natural to consider whether such approaches can be extended to the setting of multi-view networks. However, there has been little work along this direction.  

Effective approaches to multi-view network embedding using GAN have to overcome the key challenge that is absent in the single-view setting: in the single-view setting,  the generator, in order to produce adversarial samples, needs to model only the connectivity (presence or absence of a link) between pairs of nodes, in multi-view networks, how to model not only the connectivity within each of the different, but also the complex correlations between views.
The key contributions of the paper are as follows:
\begin{itemize}
\item 
We propose the \underline{M}ulti-view network \underline{E}mbedding \underline{GAN} (MEGAN), a novel GAN framework for learning a low dimensional, typically non-linear and information preserving embedding of a given multi-view network. 
\item 
Specifically, we show how to design a generator that can effectively produce adversarial data samples in the multi-view setting. 
\item
We describe an unsupervised learning algorithm for training MEGAN from multi-view network data. 
\item 
Through experiments with real-world multi-view network data, we show that MEGAN outperforms other state-of-the-art methods for MVNRL when the learned representations are used for node labeling, link prediction, and visualization.
\end{itemize}


\section{Related Work}
\subsection{Single-view Network Embedding}
Single-view network embedding methods seek to learn a low-dimensional, often non-linear and information preserving embedding of a single-view network for node classification and link prediction tasks. 
There is a growing literature on single-view network embedding methods ~\cite{perozzi2014deepwalk,tang2015line,ou2016asymmetric,wang2016structural,grover2016node2vec,hamilton2017inductive,zhang2018arbitrary}. For example, in~\cite{belkin2001laplacian}, spectral analysis is performed on Laplacian matrix and the top-$k$ eigenvectors are used as the representations of network nodes. DeepWalk~\cite{perozzi2014deepwalk} introduces the idea of Skip-gram, a word representation model in NLP, to learn node representations from random-walk sequences. SDNE~\cite{wang2016structural} uses deep neural networks to preserve the neighbors structure proximity in network embedding.  GraphSAGE~\cite{hamilton2017inductive} generates embeddings by recursively sampling and aggregating features from a node’s local neighborhood. 
GraphGAN~\cite{wang2018graphgan}, which extends GAN~\cite{goodfellow2014generative} to work with networks (as opposed to feature vectors) has shown promising results on the network embedding task. It learns a generator to approximate the node connectivity distribution and a discriminator to differentiate "fake" nodes (adversarial samples) and the nodes sampled from the true data distribution. We differ from their work in mainly two  respects: we focus on the complex multi-view network; we develop connectivity discriminator and generator with novel sampling strategy. 

\subsection{Multi-view Network Embedding}
Motivated by real-world applications, there is a growing interest in methods for learning embeddings of multi-view networks. Such methods effectively integrate information from the individual network views while exploiting complementarity of information supplied by the different views.  To capture the associations across different views, \cite{ma2017multi} utilize a tensor to model the multi-view network and  factorize tensor to obtain a  low-dimensional embedding; MVE \cite{qu2017attention} combine information from multiple views using a weighted voting scheme;  MNE~\cite{zhang2018scalable} use a latent space to integrate information across multiple views. In contrast, MEGAN proposed in this paper implicitly models the associations between views in a latent space and employs a generator that effectively integrates information about pair-wise links between nodes across all of the views.  

\section{Multi-View Network Embedding GAN}
In what follows, we define multi-view network embedding problem before describing the key components of MEGAN, our proposed solution to multi-view network embedding.
\subsection{Multi-View Network Embedding}
A multi-view network is defined as $\mathcal{G} = (\mathcal{V} ,\mathcal{E})$, where $\mathcal{V}= \{v_1, v_2, \dots, v_n \}$ denotes the set of nodes and $\mathcal{E}=\{\mathcal{E}^{(1)},\mathcal{E}^{(2)}, \dots, \mathcal{E}^{(k)}\} $ describes the edge sets that encode  $k$ different relation types (views).
For a given relation type $l$, $\mathcal{E}^{(l)} = \{e_{ij}^{(l)}\}_{i,j=1}^n$ specifies the presence of the corresponding relation between node $v_i$ and node $v_j$. Thus, $e_{ij}^{(l)}=1$ indicates that $v_i$ and $v_j$ has a link for relation $l$ and 0 otherwise.  We use $\mathcal{K}_{ij} = (e_{ij}^{(1)}, \dots, e_{ij}^{(k)})$ to denote the connectivity of $(v_i,v_j)$ in multi-view topology. The goal of multi-view network representation learning is to learn $\mathbf{X} \in \mathbb{R}^{n \times d}$, a low-dimensional embedding of $\mathcal{G}$, where $d \ll n$ is the latent dimension.
\begin{figure}[ht!]
      \centering
        \includegraphics[width=9cm]{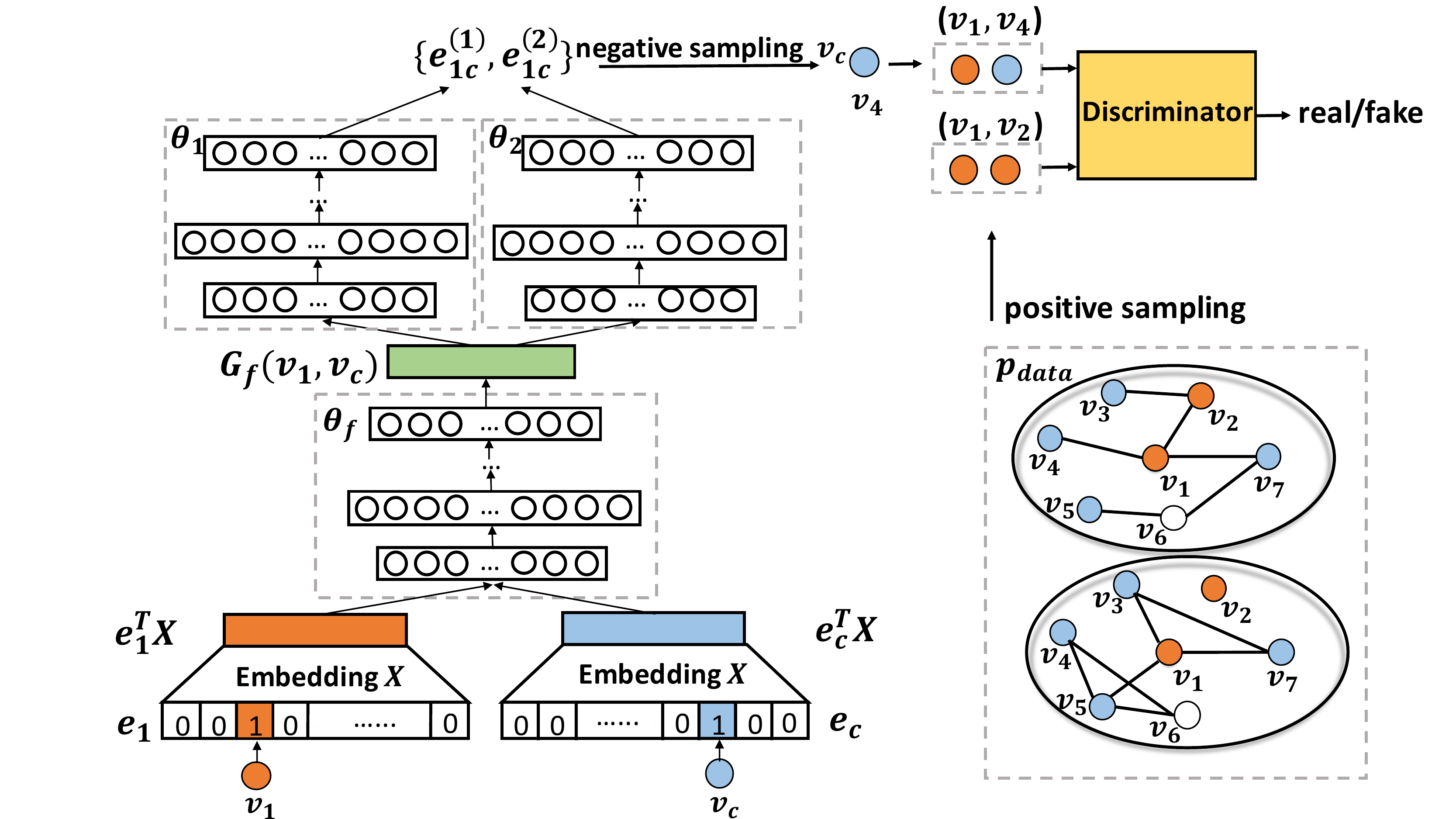}
        \caption{The architecture of MEGAN framework. The pair of the orange nodes $(v_1,v_2)$  denotes {\em real} pair of nodes. For each blue nodes $v_c$, the generator learns its connectivity $\{e^{(1)}_{ic},e^{(2)}_{ic}\}$ and generates {\em fake} pair of nodes which could fool the discriminator with highest probability. In the example, $(v_1, v_4)$ is selected to fool the discriminator.}
        \label{fig:model_architecture}
        \vskip -1em
\end{figure}


\subsection{\textbf{Multi-view Generative Adversarial Network}}
Unlike a GAN designed to perform single-view network embedding, which needs to model only the connectivity among nodes within a single view, the MEGAN needs to capture the connectivity among nodes within each view as well as the correlations between views. To achieve this, given the \textit{real} pair of nodes $(v_i, v_j) \sim p_{data}$, MEGAN consists of two modules: a \textbf{generator} which generates (or chooses) a \textit{fake} node $v_c$ with the connectivity pattern $\mathcal K_{ic}$ between $v_i$ and $v_c$ being sufficiently similar to that of the real pair of node; and a \textbf{discriminator} which is trained to distinguish between the real pair of nodes $(v_i,v_j)$ and \textit{fake} pair of node $(v_i,v_c)$.
Figure~\ref{fig:model_architecture} illustrates the architecture of MEGAN. For the pair of nodes $(v_1, v_2)$, the generator produces a \textit{fake} node $v_4$ forming the pair of nodes $(v_1, v_4)$ to fool the discriminator and  the discriminator is trained to differentiate if a pair of nodes input to it is \textit{real} or \textit{fake}. With the minmax game between $D$ and $G$, we will show that upon convergence, we are able to learn embeddings that $G$ can use to generate the multi-view network. In such situation, the learned embeddings is able to capture the connectivity among nodes within each views and the correlations between views. Next, we introduce the details of $G$, $D$ and an efficient sampling strategy.

\subsubsection{Multi-view Generator}
Recall that goals of multi-view generator $G$ are to \textbf{(i)} generate the multi-view connectivity of fake nodes that
fools the discriminator $D$; and \textbf{(ii)} learn an embedding  that captures the multi-view network topology. To ensure that it achieves the goals, we design the multi-view generator consisting of two components: one fusing generator $G_f$ and $k$ connectivity generators $G^{(l)}$. Let $\mathbf{X} \in \mathbb{R}^{n \times d}$ be the network embedding matrix we want to learn, then the representation of $v_i$ can be obtained as $\mathbf{e}_i^T \mathbf{X}$, where $\mathbf{e}_i$ is the one-hot representation with $i$-th value as 1. The fusing generator $G_f$ firstly fuses the representation of $v_i$ and $v_j$, aiming to capture the correlation between $v_i$ and $v_j$.  We use $G_f(v_i, v_j; \mathbf{X}, \theta_f)$ to denote the fused representation where $\theta_f$ is its parameters. Then, the fused representation is used to calculate the probability of $e_{ij}^{(l)}=1$ for $l=1,\dots, k$. This can be formally written as:
\begin{equation}
\small
\begin{split}
G(e^{(1)}_{ij},\dots,e^{(k)}_{ij}|v_i,v_j,\theta_G)= \prod_{l=1}^kG^{(l)}(e^{(l)}_{ij}|G_f(v_i,v_j;\mathbf{X},\theta_f);\theta^{(l)})
\end{split}
\label{eq4:g}
\end{equation}
where $G^{(l)}$ is the connectivity generator for generating the connectivity in between $(v_i,v_j)$ in view $l$ given the fused representation and $\theta^{(l)}$ is the parameter for $G^{(l)}$. $\theta_G=\{\mathbf{X}, \theta_f, \theta^{(l), l =1, \dots, k}\}$ is the parameter for the multi-view generator $G$. Because each $G^{(l)}$ of $k$ generators are independent given the fused representation, the corresponding parameter $\theta^{(l)}$ can be optimized independently in parallel. We use two layer multi-layer perceptron (MLP) to implement $G_f$ and each $G^{(l)}$. Actually, more complex deep neural networks could replace the generative model outlined here. We leave exploring feasible deep neural networks as a possible future direction. To fool the discriminator, for each $(v_i,v_j) \sim p_{data}$, where $p_{data}$ represents the multi-view connectivity, we propose to sample a pair of negative nodes from $G$ that has connectivity that is similar to that of  $(v_i, v_j)$, e.g., $(v_i, v_c) \sim p_g$ with $\tilde{\mathcal{K}}_{ic}=\mathcal{K}_{ij}$, where $p_g$ denotes the distribution modeled by $G$. In particular, we choose the $v_c$ that has the highest probability as:
\begin{equation}
\argmax_{v_c \in \mathcal V}\prod_{l=1}^k G^{(l)}(e^{(l)}_{ic} =e^{(l)}_{ij} |v_i,v_j,G_f(v_i,v_j;\mathbf{X},\theta_f)
    \label{eq:select}
\end{equation}
The motivation behind the negative sampling strategy in Eq.(\ref{eq:select}) is that the negative node pair $(v_i, v_c)$ is more likely to fool the connectivity discriminator if connectivity $\tilde{\mathcal{K}}_{ic}$ is the same to the connectivity $\mathcal{K}_{ij}$ of the positive pair $(v_i, v_j)$. \textit{The objective of $G$ is then to update network embedding $\mathbf{X}$ and $\theta_G$ so that the generator has higher chance of producing negative samples that can fool the discriminator $D$}.

\subsubsection{Node Pair Discriminator}
The goal of $D$ is to discriminate between the positive node pairs from the multi-view network data and the negative node pairs produced by the generator $G$ so as to enforce $G$ to more accurately fit the distribution of multi-view network connectivity. For this purpose, for an arbitrary node pair sampled from real-data, i.e., $(v_i,v_j) \sim p_{data}$, $D$ should output 1, meaning that the sampled node pair is real. 
Given such a negative or fake edge (pair of nodes), the discriminator should output 0, whereas $G$ should aim to assign high enough probability to negative pair of nodes that can fool $D$. 
As $D$ learns to distinguish the negative pair of nodes from the positive ones, the $G$ captures the connectivity distribution of the multi-view graph. 
We will show this in Section~\ref{sec:theorectical_analysis}.

We define the $D$ as the sigmoid function of the inner product of the input node pair $(v_i,v_j)$:
\begin{equation}
\small
   D(v_i,v_j) =  \frac{\exp(\mathbf{d_i}^T\cdot \mathbf{d_j})}{1+\exp(\mathbf{d_i}^T\cdot \mathbf{d_j})}
   \label{eq2:d}
\end{equation}
where $\mathbf{d}_i$ and $\mathbf{d}_j$ denote the $d$ dimensional representation of node pair $(v_i,v_j)$. It is worth noting that $D(v_i,v_j)$ could be any differentiable function with domain $[0,1]$. We choose the sigmoid function for its stable property and leave the selections of $D(\cdot)$ as a possible future direction.

\subsubsection{Efficient Negative Node Sampling}
In practice, for one pair of nodes $(v_i,v_j)$, in order to sample a pair of nodes $(v_i, v_c)$ from $G$, we need to calculate Eq.(\ref{eq:select}) for all $v_c \in \mathcal{V}$ and select the $v_c$ with the highest probability, which can be very time-consuming. To make the negative sampling more efficient, instead of calculating the probability for all nodes, we  calculate the probability  for neighbors of $v_i$, i.e., $\mathcal{N}(v_i)$, where $\mathcal{N}(v_i)$ is the set of nodes that have connectivity to $v_i$ for at least one view in real network. In other words, we only sample from $\mathcal{N}(v_i)$. The size of $\mathcal{N}(v_i)$ is significantly smaller than $n$ because the network is usually very sparse, which makes the sampling very efficient. 



\subsubsection{Objective Function of MEGAN}

With generator modeling multi-view connectivity to generate fake samples that could fool the discriminator, discriminator differentiates between true pairs of nodes from fake pairs of nodes. We specify the objective function for MEGAN: 
\begin{equation}
\small
\begin{split}
\min \limits_{\theta_G} &\max \limits_{\theta_D}V(G,D)=\sum_{i=1}^{n}(\mathbb{E}_{(v_i, v_j)\sim p_{\rm data}}[{\rm log}D(v_i,v_j;\theta_D)]\\
&+\mathbb{E}_{{(v_i,v_c)\sim p_g}}[{\rm log}(1-D(v_i,v_c;\theta_D))])
\end{split}
\label{ref:eqValue}
\end{equation}
where $(v_i,v_j)\sim p_{\rm data}$ denotes the positive nodes pair and $(v_i,v_c) \sim p_g$ denotes the fake pair of nodes obtained using the efficient negative sampling strategy outlined above. \textit{Through such minmax game, we can learn network embedding $\mathbf{X}$ and the generator $G$ that can approximate the multi-view network to fool the discriminator}. In other words, the learned network embedding $\mathbf{X}$ is able to capture the connectivity among nodes within each views and the correlations between views.

\subsection{Training Algorithm of MEGAN}
Following the standard approach to  training GAN ~\cite{goodfellow2014generative,wang2018graphgan}, we alternate between the updates $D$ and $G$  with mini-batch gradient descent. \\
\noindent{}\textbf{Updating $D$:}
Given that $\theta_D = \{ d_i \}_{i=1}^n$ is differentiable w.r.t to the loss function in Eq.(\ref{ref:eqValue}), the gradient of $\theta_D$ is given as:
\begin{equation}
\small
\begin{split}
    &\nabla_{\theta_D}V(G,D) =\sum_{i=1}^{n}(\mathbb{E}_{(v_i, v_j)\sim p_{\rm data}}[\nabla_{\theta_D}{\rm log}D(v_i,v_j)]\\
&+\mathbb{E}_{{(v_i,v_c)\sim p_g}}[\nabla \theta_D{\rm log}(1-D(v_i,v_c))])
\end{split}
\label{eq3:updateD}
\end{equation}
\noindent{}\textbf{Updating $G$:}
Since we sampled discrete data, i.e., the negative node IDs, from MEGAN, the discrete outputs make it difficult to pass the gradient update from the discriminative model to the generative model. Following the previous work~\cite{yu2017seqgan,wang2018graphgan}, we utilize the policy gradient to update the generator parameters $\theta_G = \{\mathbf{X}, \theta_f, \theta^{(1)}, \dots, \theta^{(k)} \}$:
\begin{equation}
\small
\begin{split}
    &\nabla_{\theta_G}V(G,D)\\ &=\nabla_{\theta_G} \sum_{i=1}^{n}\mathbb{E}_{(v_i,v_c)\sim p_g}[{\rm log}(1-D(v_i,v_c;\theta_D)]\\
     &=\nabla_{\theta_G}\sum_{i=1}^{n}\sum_{c=1}^{n} G(\mathcal K_{ic}|v_i,v_c)[{\rm log}(1-D(v_i,v_c;\theta_D)]\\
        &=\sum_{i=1}^{n}\mathbb{E}_{(v_i,v_c)\sim G}[\nabla_{\theta_G}{\rm log}G(\mathcal K_{ic}|v_i,v_c){\rm log}(1-D(v_i,v_c;\theta_D))]
\end{split}
\label{eq3:updateG}
\end{equation}

\noindent{}\textbf{Training Algorithm}
With the update rules for $\theta_D$ and $\theta_G$ in place, the overall training algorithm is summarized in Algorithm \ref{Multi}. 
In Line 1, we initialize and pre-train the $D$ and $G$. From Line 3 to 6, we update parameters of $G$, i.e., $\theta_G$. From Line 7 to 10, we update the parameters of $D$, i.e., $\theta_D$ 
The $D$ and $G$ play against each other until the MEGAN converges.

\subsection{Theoretical Analysis} \label{sec:theorectical_analysis}
It has been shown in ~\cite{goodfellow2014generative} that $p_g$ could converge to $p_{data}$ in continuous space. In Proposition \ref{prop:1}, we show that $p_g$ also converge $p_{data}$ in discrete space. In other words, upon convergence of Algorithm~\ref{Multi}, MEGAN can learn embeddings that makes $p_g \approx p_{data}$, which captures the multi-view network topology. 

\begin{algorithm}[t]
  \caption{\text{MVGAN} framework}
  \label{Multi}
  \begin{algorithmic}[1]
  \REQUIRE embedding dimension $d$, size of discriminating samples $t$, generating samples $s$
  \ENSURE $\theta_D$, $\theta_G$
 \STATE Initialize and pre-train $D$ and $G$  
  \WHILE{MVGAN not converged} 
  \FOR {G-steps}
    \STATE  Sample $s$ negative pairs of nodes $(v_i, v_c)$ for the given   positive pair of nodes $(v_i,v_j)$
\STATE update $\theta_G$ according to Eq.(\ref{eq4:g}) and Eq.(\ref{eq3:updateG})
    \ENDFOR
    \FOR {D-steps}
    \STATE  Sample $t$ positive nodes pairs $(v_i, v_j)$ and $t$ negative
node pairs $(v_i, v_c)$ from $p_g$ for each node $v_i$
\STATE update $\theta_D$ according to Eq.(\ref{eq2:d}) and Eq.(\ref{eq3:updateD})
    \ENDFOR
  \ENDWHILE
  \end{algorithmic}
\end{algorithm}
\vskip -1.1em
\begin{proposition}
\label{prop:1}
If $G$ and $D$ have enough capacity, the discriminator and the generator are allowed to reach its optimum, and $p_g$ is converge to $p_{data}$ based on the update rule in Algorithm \ref{Multi}.
\end{proposition}
\begin{proof}
The proof is similar to ~\cite{goodfellow2014generative}, and we omit the details here.
\end{proof}
\vskip -1.1em

\vskip -2em
\section{Experiments}
We report results of our experiments with two benchmark multi-view network data sets designed to compare the quality of multi-view embeddings learned by MEGAN and other state-of-the-art network embedding methods. We use the  embedding learned by each method on three tasks, namely, node classification, link prediction, and network visualization. In these tasks, we use the performance as a quantitative proxy measure for the quality of the embedding. We also examine sensitivity of MEGAN w.r.t the choice of hyperparameters. 

\subsection{Data Sets}
We use the following multi-view network data sets\cite{bui2016labeling} in our experiments: (i). {\bf Last.fm}: Last.fm data were collected from the online music network Last.fm\footnote{https://www.last.fm}. The nodes in the network represent users of Last.fm and the edges denote different types of relationships between users, e.g., shared interest in an artist, event, etc. (ii). {\bf Flickr:} Flickr data were collected from the Flickr photo sharing service. The views correspond to different aspects of shared interest between users in photos (e.g., tags, comments, etc.).  The statistics of the data sets are summarized in Table \ref{tab:Statistical}. The number of edges denotes the total number of edges (summed over all of the views). Perhaps not unsurprisingly, the degree distribution of each view of the data approximately follows power-law degree distribution.

\begin{table}[ht!]
\centering
  \label{tab:Statistical}
  \begin{tabular}{cccccc}
    \toprule
    Datasets&\#nodes&\#edges &\#view&\#label \\
    \midrule
    Last.fm & 10,197& 1,325,367&12&11 \\
    Flickr &6,163& 378,547&5 &10\\
  \bottomrule
\end{tabular}
\caption{Summary of Flickr and Last.fm data}
\vskip -1.1em
\end{table}

\subsection{Experiments}
We compare MEGAN with the state-of-the-art single view as well as multi-view network embedding methods.  To apply single-view method, we generate a single-view network from the multi-view network by placing an edge between a pair of nodes if they are linked by an edge in at least one of the views. The single view methods included in the comparison are:
\begin{itemize}[leftmargin=*]
\item node2vec~\cite{grover2016node2vec}, a single view network embedding method, which learns network embedding that maximizes the likelihood of preserving network neighborhoods of nodes.
\item GraphGAN~\cite{wang2018graphgan}, which is a variant of  GAN for learning single-view network embedding.
\item DRNE~\cite{tu2018deep}, which constructs utilizes an LSTM to recursively aggregate the representations of node neighborhoods.
\end{itemize}
The multi-view methods included in the comparison are:
\begin{itemize}[leftmargin=*]    
    \item MNE~\cite{zhang2018scalable}, which jointly learns view-specific embeddings and an embedding that is common to all views with the latter providing a conduit for sharing information across views. 
    \item MVE~\cite{qu2017attention}, which constructs a multi-view network embedding as a weighted combination of the constituent single view embeddings.
\end{itemize}
In each case, the hyperparameters were set according to the suggestions of the authors of the respective methods.  The embedding dimension was set to 128 in all of our experiments.
\vskip -0.5em
\begin{figure}[t!]
    \centering
    \includegraphics[width=8.5cm]{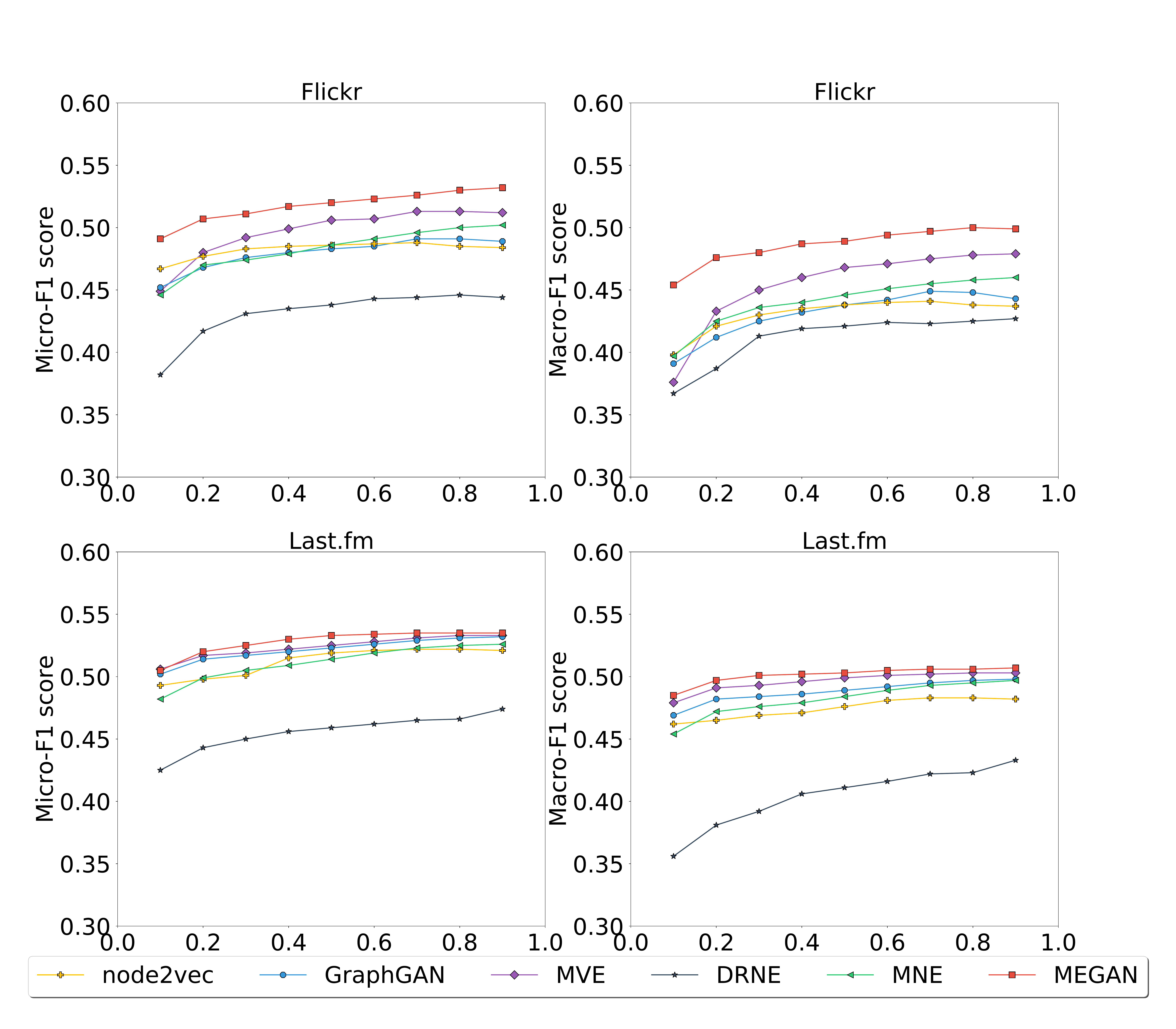}
    \vskip -1em
    \caption{Performance comparison of node classification tasks on Flickr and Last.fm as a function of the fraction of nodes used for training the node classifier}
    \label{Fig_vis:nc}
    \vskip -1em
\end{figure}

\begin{figure}[t!]
      \centering
    \subfigure[Flickr]{
      \includegraphics[width=8.6cm]{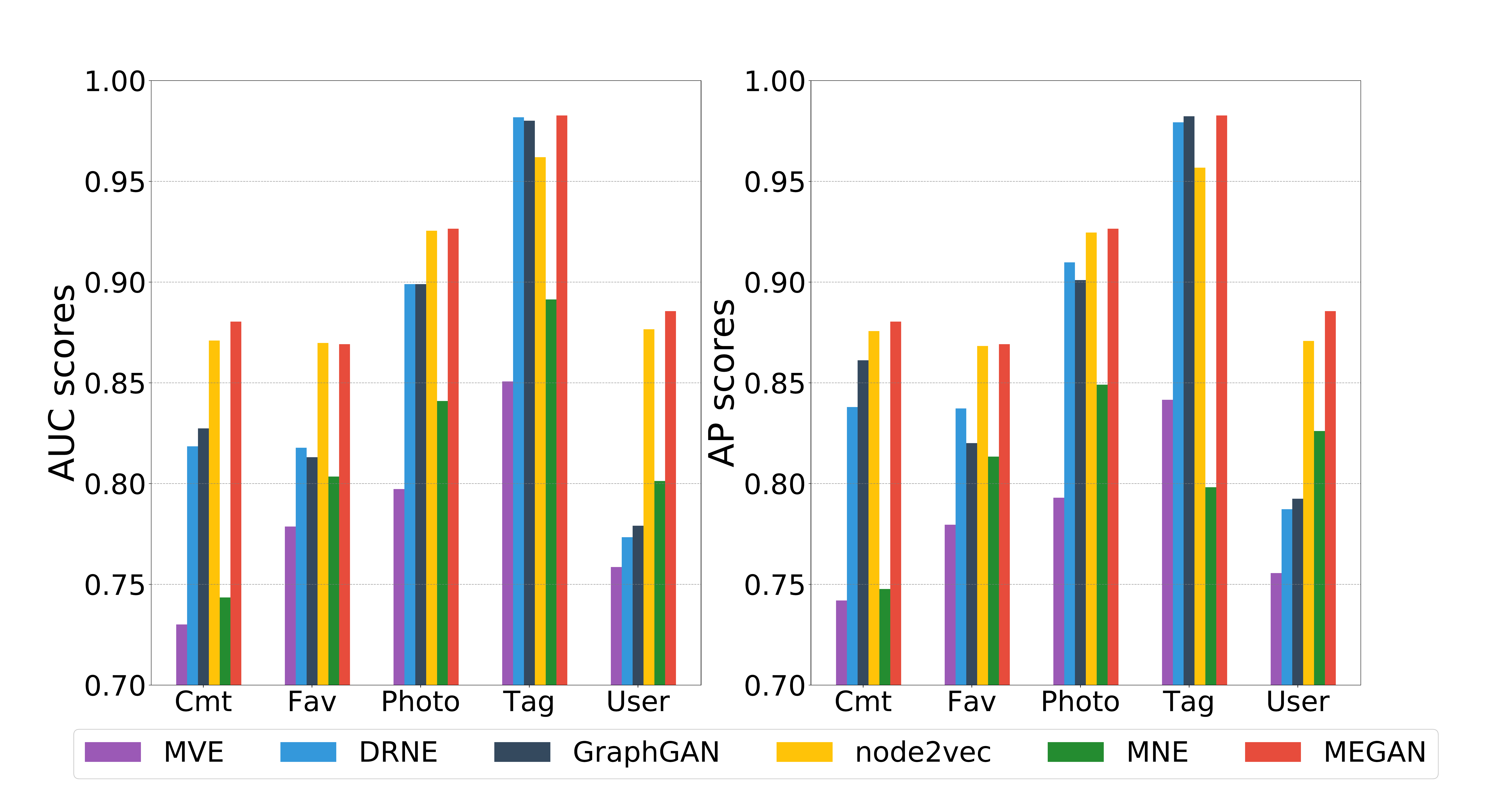}
     \label{Fig_vis:lfFlickr}}
     \vskip  -0.5em
     
    \subfigure[Last.fm]{
        \includegraphics[width=8.6cm]{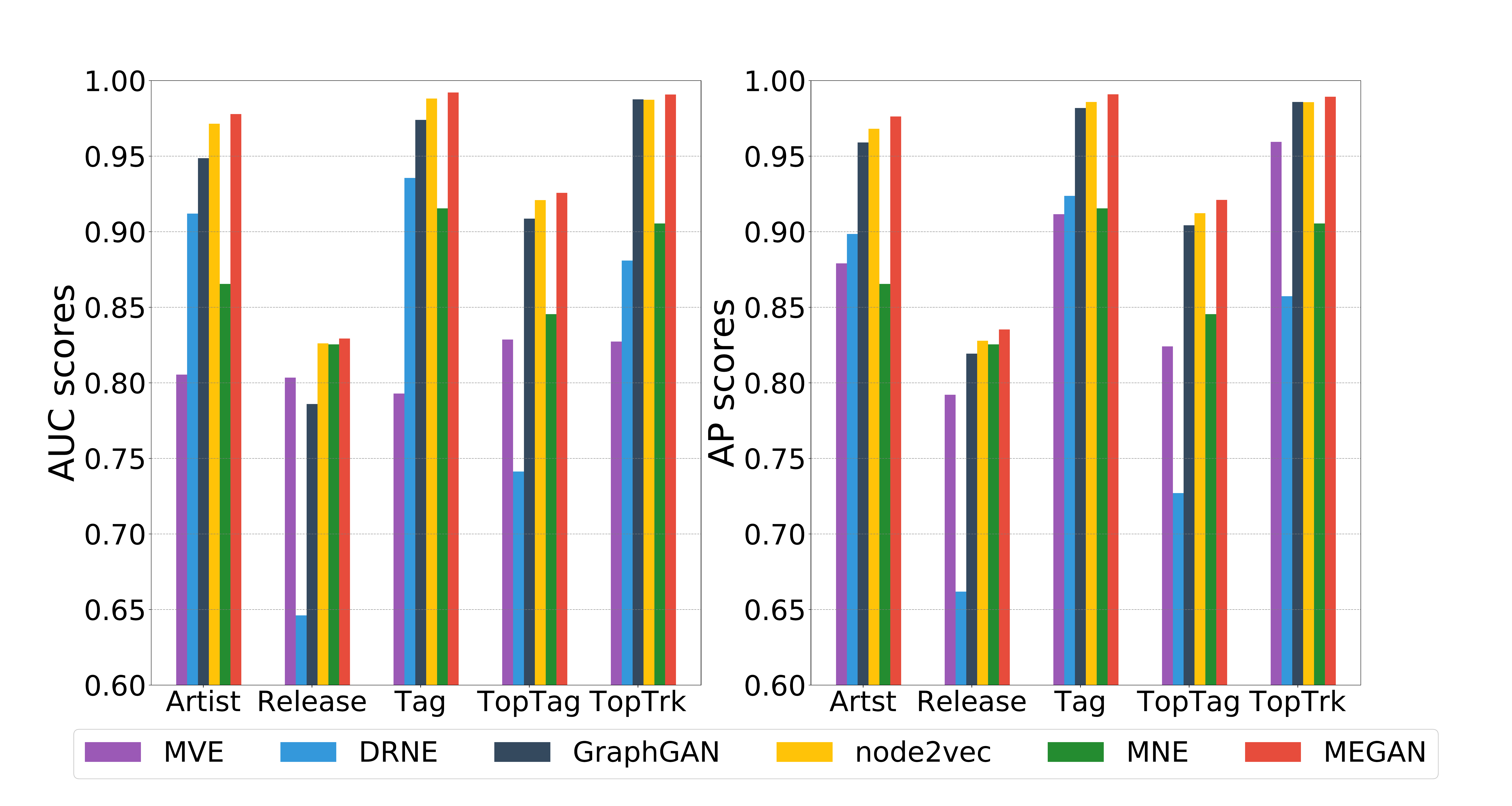}
        \label{Fig_vis:lkLast}}
       \vskip  -1em
     \caption{Performance comparison of link prediction tasks}
     \vskip -1em
\end{figure}
\begin{figure*}[t!]
    \centering
    \subfigure[MEGAN]{
 \includegraphics[width=2.5cm]{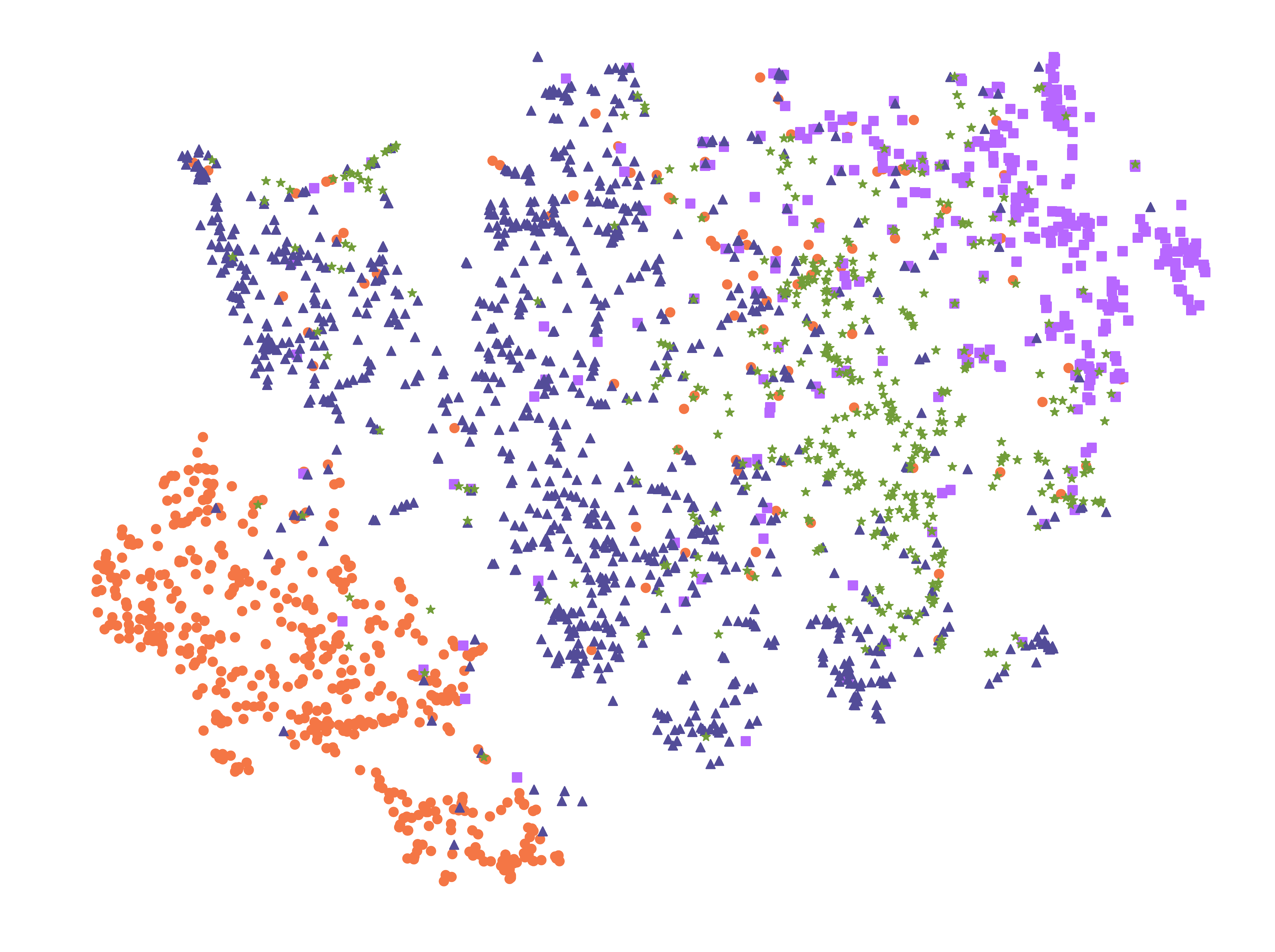}
        \label{Fig_vis:MVGAN}%
    }
    \subfigure[Node2Vec]{ 
        \includegraphics[width=2.5cm]{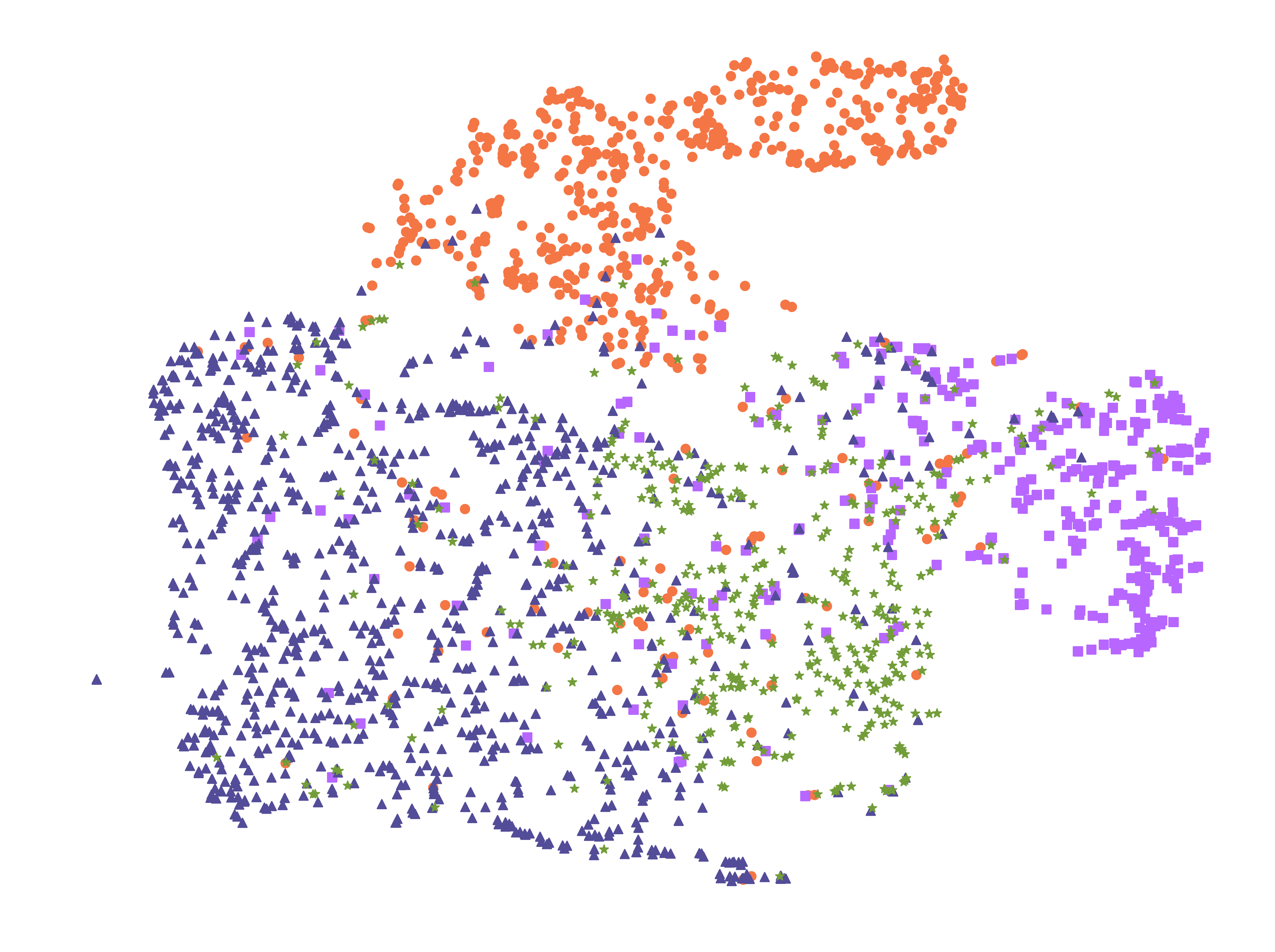}
        \label{Fig_vis:Node2Vec}%
    }
     \subfigure[MVE]{ 
        \includegraphics[width=2.5cm]{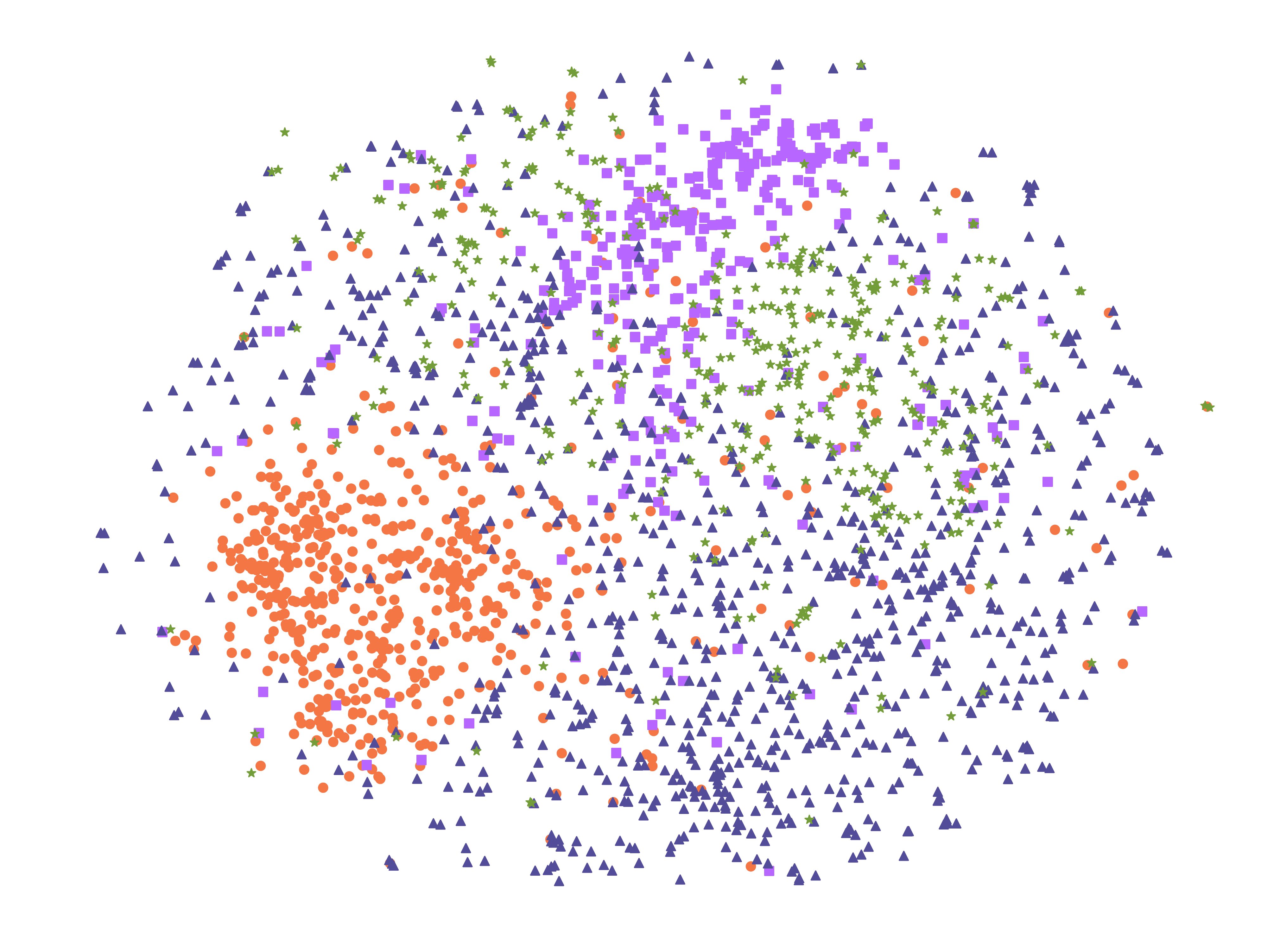}
        \label{Fig_vis:MVE}%
    }
    \subfigure[GraphGAN]{
        \includegraphics[width=2.5cm]{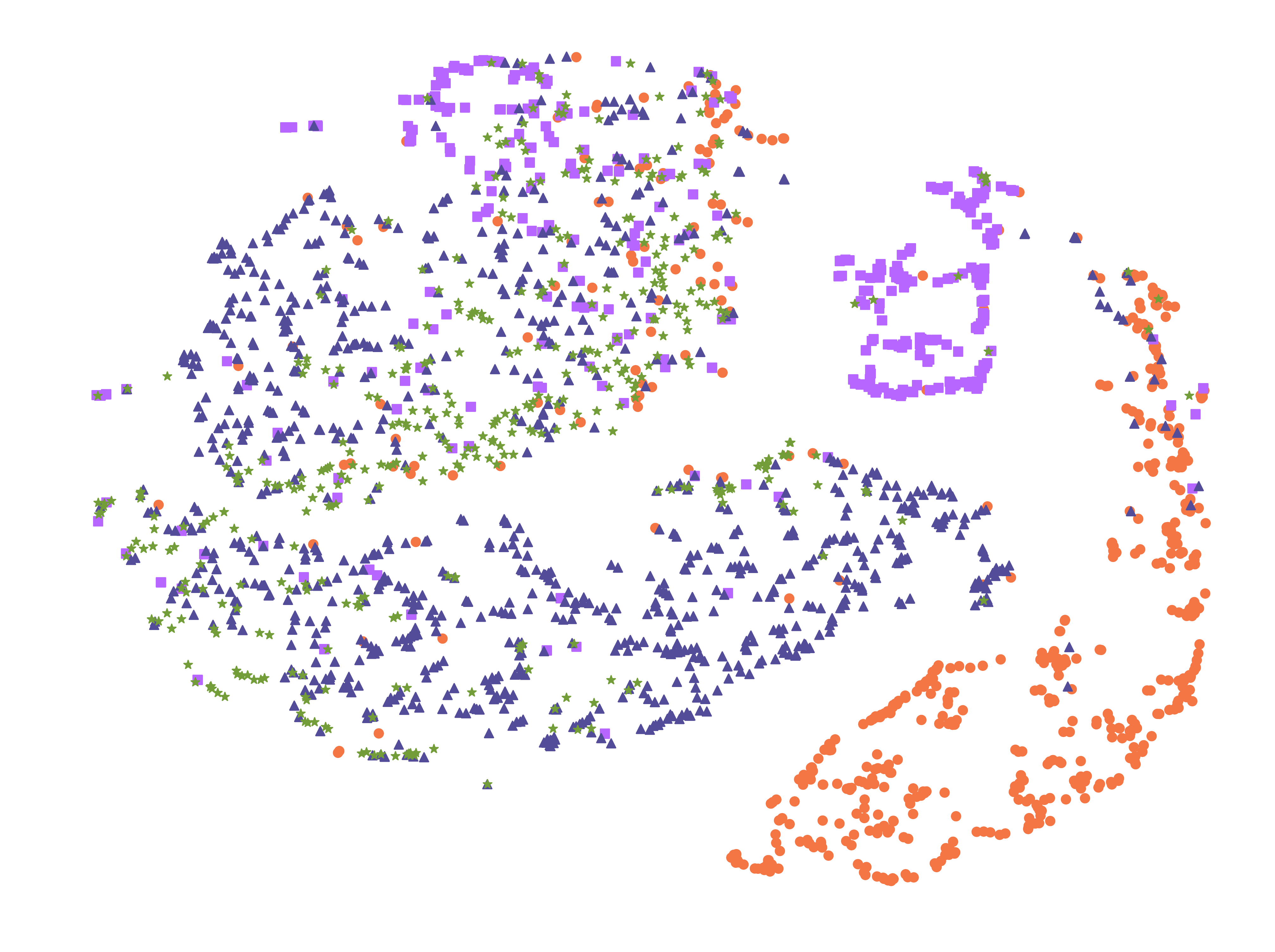}
        \label{Fig_vis:GraphGAN}%
        }
        \subfigure[DRNE]{
        \includegraphics[width=2.5cm]{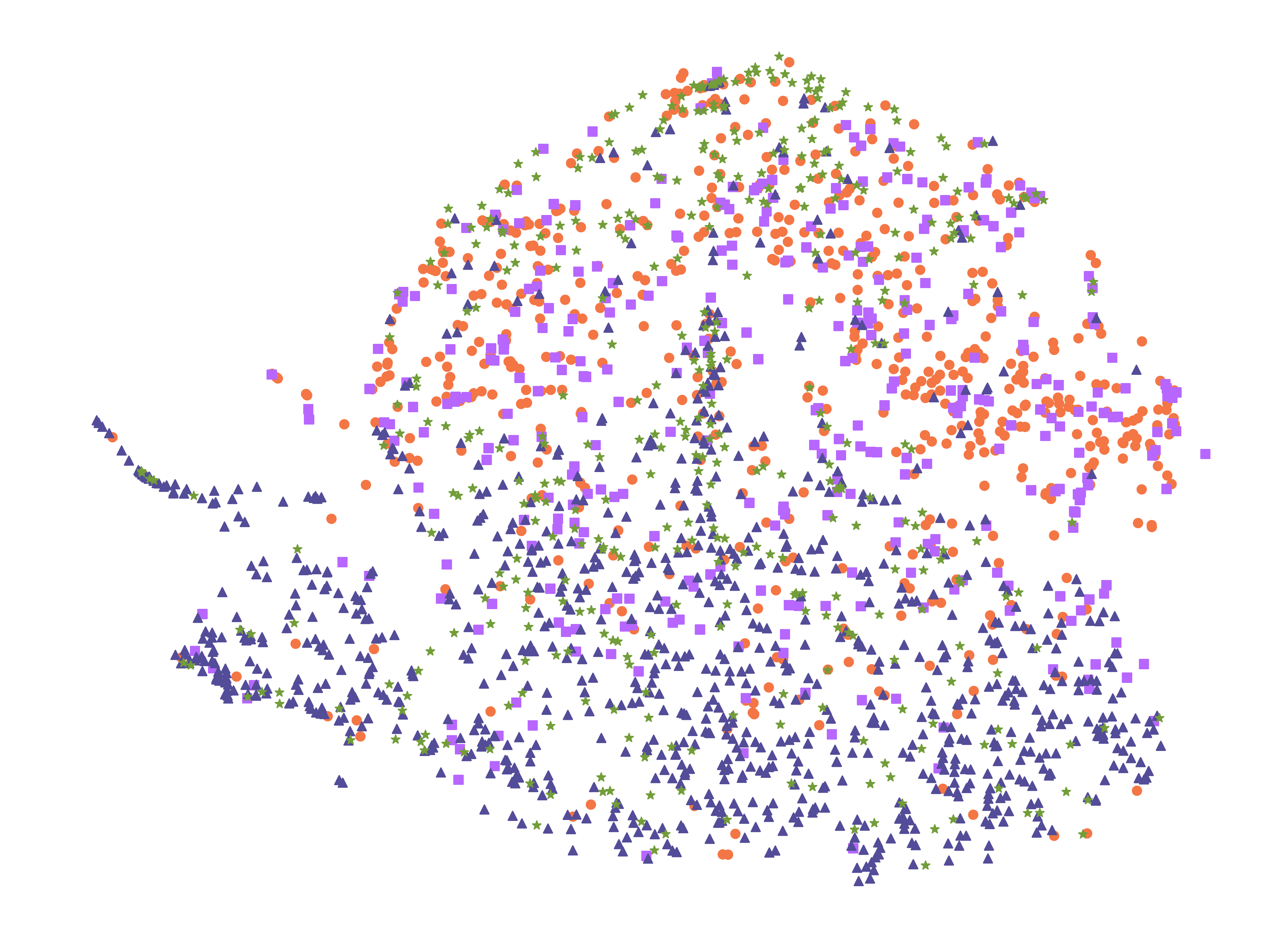}
        \label{Fig_vis:DRNE}}%
        \subfigure[MNE]{
        \includegraphics[width=2.5cm]{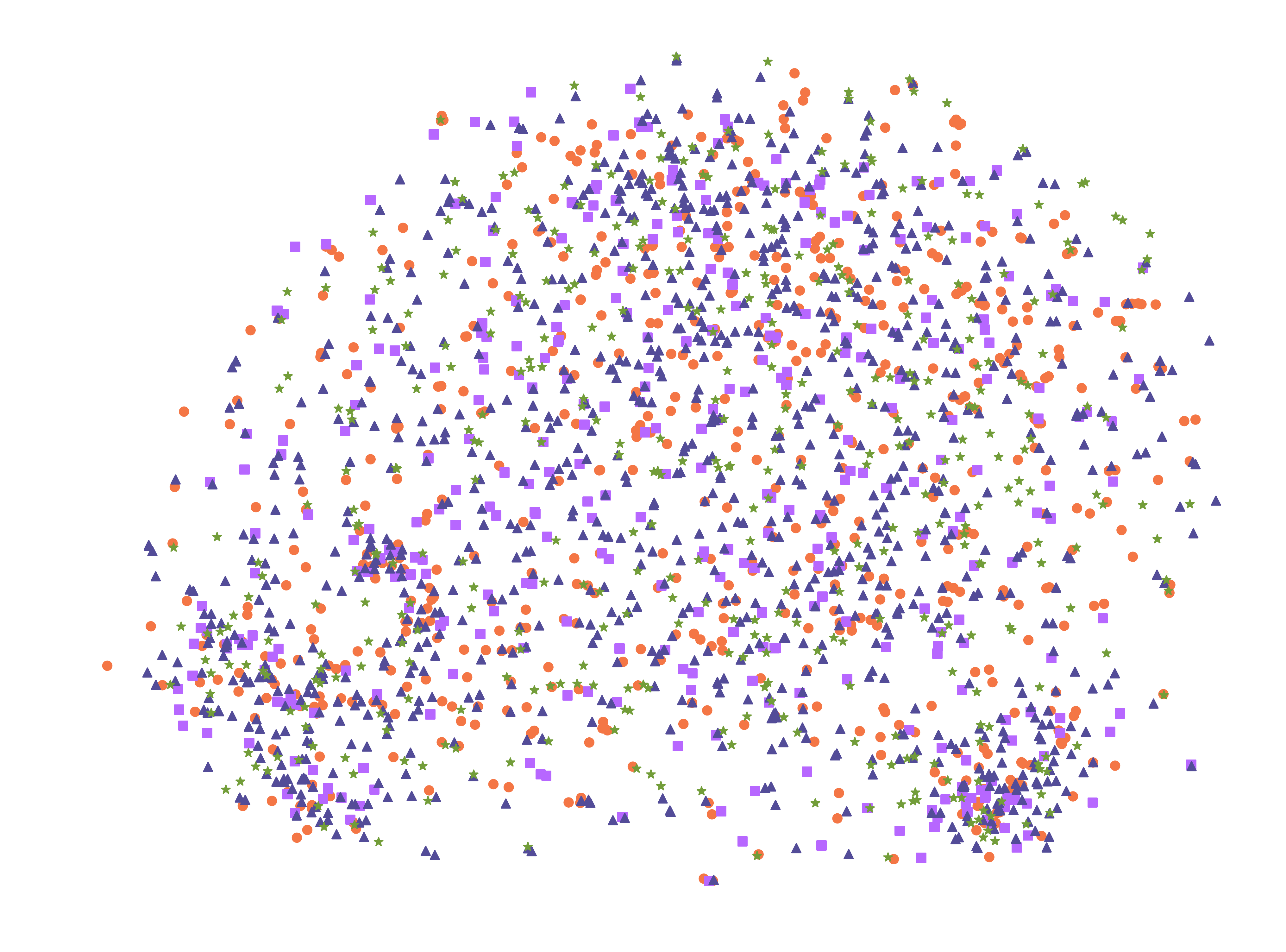}
        \label{Fig_vis:MNE}}%
    \vskip -0.5em
    \caption{Visualization results of Flickr. The users are projected into 2D space. Color of the node represents categories of users.}
    \vskip -1em
    \label{fig:Fig_multicore}
\end{figure*}
\subsection{Results} 
\vskip -0.2em
\subsubsection{Node Classification}\label{sec:node_classification}
We report results of experiments using the node representations  produced by each of the network embedding methods included in our comparison on the transductive node classification task. 
In each case, the network embedding is learned in an unsupervised fashion from the available multi-view network data without making use of the node labels. We randomly select $x$ fraction of the nodes as training  data (with the associated node labels added) and the remaining $(1-x)$ fraction of the nodes for testing. We run the experiments for different choices of $x \in \{0.1, 0.2, \dots, 0.9\}$. In each case, we train a standard one-versus-rest L2-regularized logistic regression classifier on the training data and evaluate its performance on the test data. We report the  performance of the node classification  using  the  Micro-F1 and Macro-F1 scores  averaged over the 10 runs for each choice of $x$ in Fig.~\ref{Fig_vis:nc}. 

Our experiments results show that: \textbf{(i)} The single view methods, GraphGAN and Node2vec achieve comparable performance; \textbf{(ii)} Multi-view  methods, MVE and MVGAN outperform the single-view methods. \textbf{(iii)}  MEGAN outperforms all of the other methods on both data sets. These results further show that multi-view methods that construct embeddings that incorporate complementary information from all of the views outperform those that do not. GAN framework offers the additional advantage of robustness and improved generalization that comes from the use of adversarial samples.  

\subsubsection{Link Prediction}
We report results of experiments using the node representations  produced by each of the network embedding methods included in our comparison on the link prediction task. Given a multi-view network, we randomly select a view, and randomly remove  50\% of the edges present in that view. We then train a classifier on the remaining data to predict the links that were removed.   Following \cite{grover2016node2vec}, we cast the link prediction task as the binary classification problem  with the network edges that were not removed used as positive examples, and an equal number of  randomly generated edges that do not appear in the network as negative examples. We represent links using embeddings of the corresponding pair of nodes.  We train and test a random forest classifier for link prediction.
In the case of the Flickr dataset, we repeat the above procedure on each of the five views; and in the case of the Last.fm data set, present the  results on  five of the most populous views. We report the  area under curve (AUC) and average precision (AP) for link prediction for the Flickr and Last.fm data sets in Fig.\ref{Fig_vis:lfFlickr} and Fig.\ref{Fig_vis:lkLast}, respectively. Based on these results, we make the following observations: \textbf{(i)} There is fairly large variability in the performance of link prediction across the different views. Such phenomenon reveals the differences in the reliability of the link structures in each view; \textbf{(ii)} In some views that are fairly rich in  links, e.g., Tag view of the Flickr data,  single view methods, such as GraphGAN and DRNE, outperform multi-view methods, MVE and MNE, and approaching MEGAN. This suggests that although single view methods may be competitive with multi-view methods when almost all of the information needed for reliable link prediction is available in  a single view, multi-view methods outperform single-view methods when views other than the target view provide complementary information for link prediction in the target view.  and \textbf{(iii)} The MEGAN outperforms its single-view counterpart GraphGAN, suggesting that the MEGAN is able to  effectively integrate complementary information from multiple views into a compact embedding.
\vskip -1em
\subsubsection{Network Visualization}
To better understand the intrinsic structure of the learned network embedding~\cite{tang2016visualizing} and reveal the quality of it, we visualize the network by projecting the embeddings onto a 2-dimensional space. We show the resulting network visualizations on the Flickr network using each of the comparison embedding methods with the t-SNE package\cite{maaten2008visualizing} in Fig.~\ref{fig:Fig_multicore}. In the figure, each color denotes one category of users' interests and we show the results for four of ten categories. Visually, in Fig.\ref{Fig_vis:MVGAN}, we find that the results obtained by using MEGAN appear to yield tighter clusters for the 4 categories with more pronounced separation between clusters as compared to the other methods .  
\begin{figure}[t]
      \centering
    \subfigure[Flickr]{
        \includegraphics[width=4.3cm]{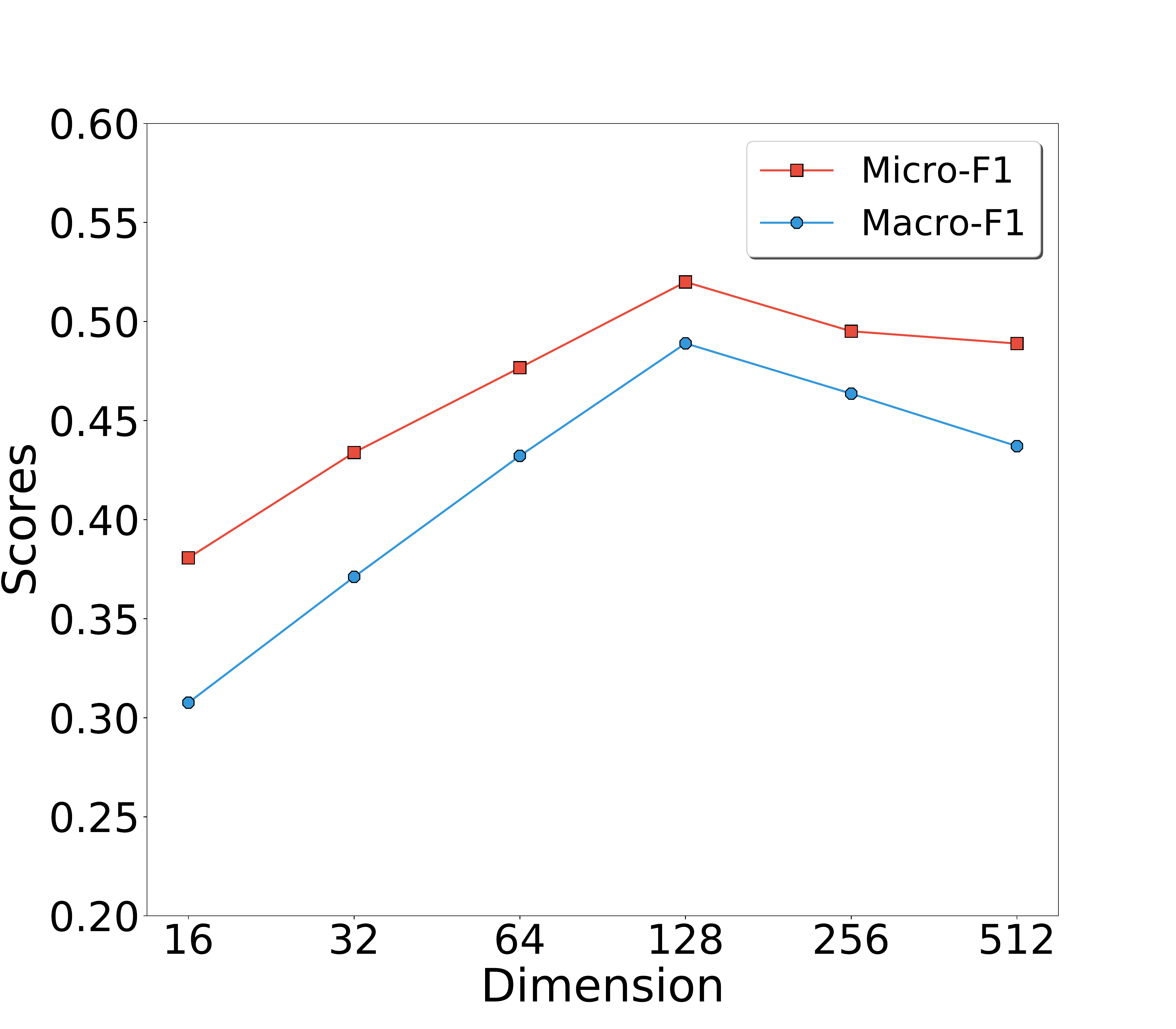}
        \label{Fig_vis:1}}
        \hskip -1.4em
     \subfigure[Last]{
      \includegraphics[width=4.3cm]{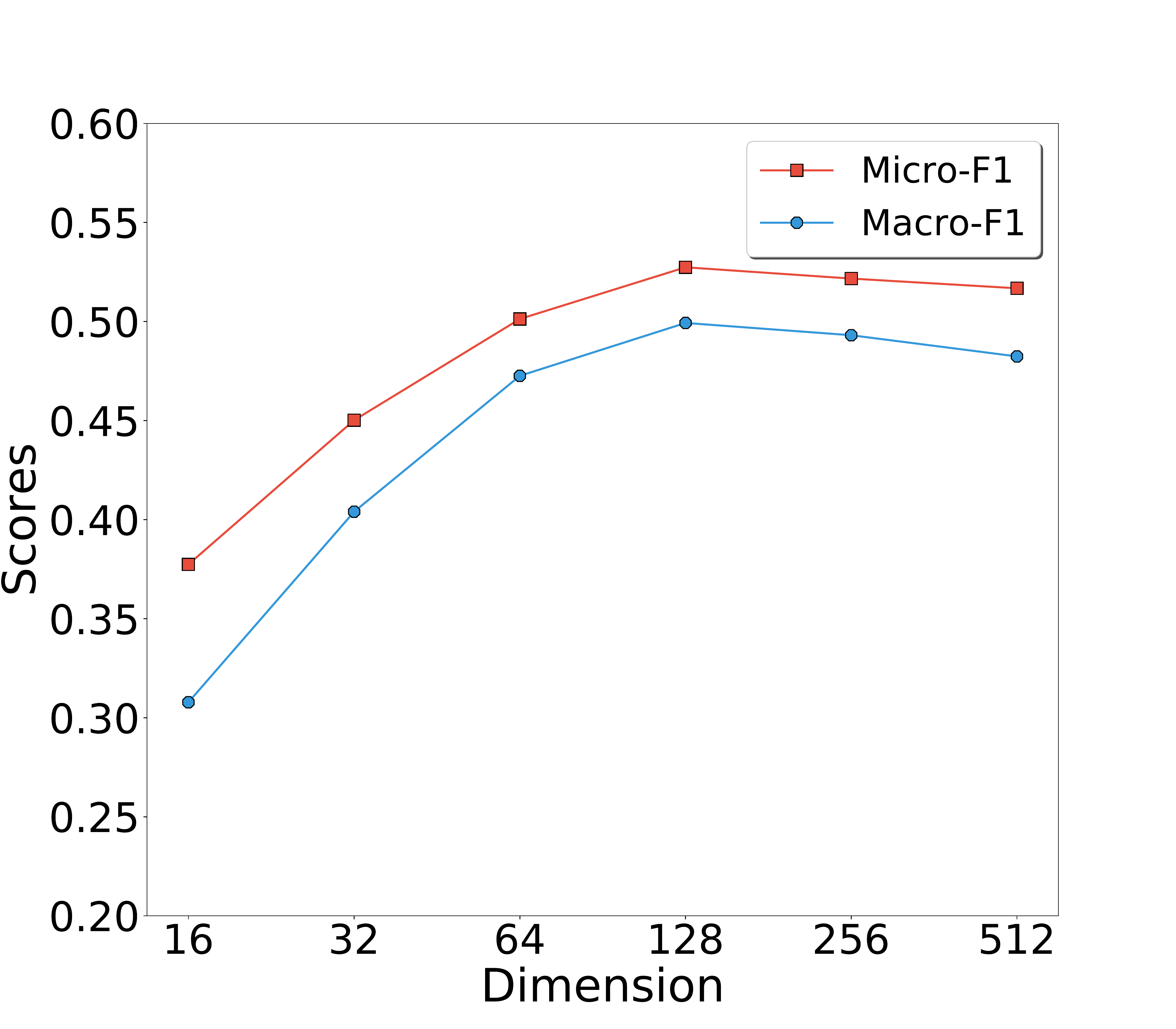}
     \label{Fig_vis:2}}
     \vskip -1em
     \caption{Parameter sensitivity for node classification on (a) Flickr and (b) Last.fm datasets with varying dimensions d.}
     \vskip -1em
\end{figure}
\subsubsection*{Impact of Embedding Dimension}
We report results of the choice of  embedding dimension on the  performance of MEGAN.  We chose  $50\%$ of the nodes randomly for training and the remaining for testing. We used different choices of the dimension $d \in \{16, 32, 64, 128, 256, 512\}$. We report the performance achieved using resulting embeddings on the node classification task using Micro-F1 and Macro-F1 for Flickr and Last.fm in  Fig.\ref{Fig_vis:1} and Fig.\ref{Fig_vis:2}. On these two data sets, we find that MEGAN achieves its optimal performance when $d$ is set to 128. This suggests that in specific applications, it may be advisable to select an optimal $d$ using cross-validation.

\vskip -1.5em


\section{Conclusions}
In this paper, we have considered the multi-view network representation learning problem, which targeting at construct the low-dimensional, information preserving and non-linear representations of {\em multi-view} networks.  Specifically, we have introduced MEGAN, a novel generative adversarial network (GAN) framework for multi-view network embedding aimed at preserving the connectivity within each individual network views, while accounting for the associations across different views. The results of our experiments with several multi-view data sets show that the embeddings obtained using MEGAN outperform the state-of-the-art methods on  node classification, link prediction and network visualization tasks. 

\section*{{\bf Acknowledgements}} This work was funded in part by grants from the  NIH NCATS through the grant UL1 TR002014 and by the NSF through the grants 1518732, 1640834, and 1636795, the Edward Frymoyer Endowed Professorship in Information Sciences and Technology at Pennsylvania State University and the Sudha Murty Distinguished Visiting Chair in Neurocomputing and Data Science funded by the Pratiksha Trust at the Indian Institute of Science (both held by Vasant Honavar). The content is solely the responsibility of the authors and does not necessarily represent the official views of the sponsors.
\newpage
\normalsize

\end{document}